\documentclass[11pt]{article}
\usepackage[english]{babel}
\usepackage[cm]{fullpage}			
\usepackage[T1]{fontenc}  			
\usepackage[utf8]{inputenc}
\usepackage{graphicx}
\usepackage[dvipsnames]{xcolor}
\usepackage{caption}
\usepackage{subcaption}
\usepackage{microtype}			 	
\usepackage{setspace} 				
\usepackage{tocbibind}				
\usepackage{mathtools}				
\usepackage{amssymb}
\usepackage{amsfonts}
\usepackage{amsthm}
\usepackage{bm} 					
\usepackage{booktabs}				
\usepackage{array}					%
\usepackage{dcolumn} 				
\usepackage[shortcuts]{extdash} 	
\usepackage{cleveref}				
\usepackage{siunitx}				
\usepackage[square, sort, numbers, authoryear]{natbib}
\usepackage{pgfplots}

\theoremstyle{theorem}
\newtheorem{thm}{Theorem}
\newtheorem{lem}[thm]{Lemma}
\newtheorem{cor}[thm]{Corollary}
\theoremstyle{definition}

\numberwithin{equation}{section}
\numberwithin{defn}{section}
\numberwithin{thm}{section}

\usepackage{physics}				
\usepackage{xparse}					

\DeclareDocumentCommand \vb 		{ m }{ \bm{\mathbf{\lowercase{ #1 }}} }
\DeclareDocumentCommand \mb 		{ m }{ \bm{\mathbf{\uppercase{ #1 }}} }
\DeclareDocumentCommand \integral 	{ O{} O{} }{ \int\limits_{#1}^{#2}\! }
\DeclareDocumentCommand \d 			{ O{x} }{\,\mathrm{d}#1}
\DeclareDocumentCommand \T 			{ }{ ^{\top} }
\DeclareDocumentCommand \I 			{ }{ ^{-1} }
\DeclareDocumentCommand \diag 		{ m }{ \mathrm{diag}\pqty{\, #1 \,} }
\DeclareDocumentCommand \eye 		{ }{ \mb{I} }
\DeclareDocumentCommand \zeroVec	{  }{ \vb{0} }
\DeclareDocumentCommand \N 			{ s O{\vb{x}} O{\vb{0}} O{\eye} }
{
	\IfBooleanTF{#1}{
		#2\sim\mathrm{N}\pqty{ #3,\, #4 }
	}{
		\mathrm{N}\pqty{\left. #2 \,\middle\vert\, #3,\, #4 \right.}
	}
}
\DeclareDocumentCommand \St 	{ O{\vb{x}} O{\vb{0}} O{\eye} O{\nu} }{ \mathrm{St}\pqty{#1 \mid #2,\, #3, #4} }
\DeclareDocumentCommand \G 		{ O{\vb{x}} O{k} O{\theta} }{ \mathrm{G}\pqty{#1 \mid #2,\, #3} }
\DeclareDocumentCommand \IG 	{ O{\vb{x}} O{k} O{\theta} }{ \mathrm{IG}\pqty{#1 \mid #2,\, #3} }
\DeclareDocumentCommand \E 		{ O{} l m }{ \mathbb{E}_{#1}\bqty#2{ #3 } }
\DeclareDocumentCommand \V 		{ O{} l m }{ \mathbb{V}_{#1}\bqty#2{ #3 } }
\DeclareDocumentCommand \Cov 	{ O{} l m }{ \mathbb{C}_{#1}\bqty#2{ #3 } }
\DeclareDocumentCommand \nlf 		{ s }{ \IfBooleanTF{#1}{g}{\vb{g}} }
\DeclareDocumentCommand \dynf 		{  }{ \vb{f} }
\DeclareDocumentCommand \obsf 		{  }{ \vb{h} }
\DeclareDocumentCommand \stVar 		{  }{ \vb{x} }
\DeclareDocumentCommand \obsVar		{  }{ \vb{z} }
\DeclareDocumentCommand \stNoise	{ s }{ \IfBooleanTF{#1}{q}{\vb{q}} }
\DeclareDocumentCommand \obsNoise	{ s }{ \IfBooleanTF{#1}{r}{\vb{r}} }
\DeclareDocumentCommand \stNoiseCov	{  }{ \mb{Q} }
\DeclareDocumentCommand \obsNoiseCov{  }{ \mb{R} }
\DeclareDocumentCommand \stMean		{ s }{ \vb{m}^x }
\DeclareDocumentCommand \obsMean	{ s }{ \vb{m}^z }
\DeclareDocumentCommand \stCov		{ s }{ \mb{P}^x }
\DeclareDocumentCommand \stObsCov	{ s }{ \mb{P}^{xz}}
\DeclareDocumentCommand \obsCov		{ s }{ \mb{P}^z }
\DeclareDocumentCommand \inVarUnt	{  }{ \vb{\xi} }
\DeclareDocumentCommand \inVar 		{  }{ {\vb{x}} }
\DeclareDocumentCommand \inMean 	{  }{ \vb{m} }
\DeclareDocumentCommand \inCov 		{  }{ \mb{P} }
\DeclareDocumentCommand \inCovFct	{  }{ \mb{L} }
\DeclareDocumentCommand \inDim 		{  }{ D }
\DeclareDocumentCommand \outVar 	{ s }{ \IfBooleanTF{#1}{y}{\vb{y}} }
\DeclareDocumentCommand \outMean 	{  }{ \vb{\mu} }
\DeclareDocumentCommand \outCov 	{  }{ \mb{\Pi} }
\DeclareDocumentCommand \outDim 	{  }{ E }
\DeclareDocumentCommand \inoutCov 	{  }{ \mb{C} }
\DeclareDocumentCommand \outMeanApp	{  }{ \vb{\mu}_{\mathrm{G}} }
\DeclareDocumentCommand \outCovApp 	{  }{ \mb{\Pi}_{\mathrm{G}} }
\DeclareDocumentCommand \inoutCovApp{  }{ \mb{C}_{\mathrm{G}} }
\DeclareDocumentCommand \R 			{  }{ \mathbb{R} }
\DeclareDocumentCommand \D 			{  }{ \mathcal{D} }
\DeclareDocumentCommand \wm 		{  }{ \vb{w} }
\DeclareDocumentCommand \wc 		{  }{ \mb{W} }
\DeclareDocumentCommand \wcc 		{  }{ \mb{W}_c }
\DeclareDocumentCommand \gpExpVar	{  }{ \sigma^2 }
\DeclareDocumentCommand \gpObs		{  }{ \vb{y} }
\DeclareDocumentCommand \kerMean	{  }{ \vb{q} }
\DeclareDocumentCommand \kerCov		{  }{ \mb{Q} }
\DeclareDocumentCommand \kerCCov	{  }{ \mb{R} }
\DeclareDocumentCommand \kerMat		{  }{ \mb{K} }

\DeclareDocumentCommand \kerf		{  }{ k }
\DeclareDocumentCommand \rbfLam		{  }{ \mb{\Lambda} }
\DeclareDocumentCommand \rbfScale   {  }{ \alpha }
\DeclareDocumentCommand \trNum		{  }{ N }

\title{Gaussian Process Quadrature Moment Transform}
\author{Jakub Prüher, Ondřej Straka}

\begin{document}
\maketitle

\begin{abstract}
	Computation of moments of transformed random variables is a problem appearing in many engineering applications. 
	The current methods for moment transformation are mostly based on the classical quadrature rules which cannot account for the approximation errors.
	Our aim is to design a method for moment transformation for Gaussian random variables which accounts for the error in the numerically computed mean. 
	We employ an instance of Bayesian quadrature, called Gaussian process quadrature (GPQ), which allows us to treat the integral itself as a random variable, where the integral variance informs about the incurred integration error.
	Experiments on the coordinate transformation and nonlinear filtering examples show that the proposed GPQ moment transform performs better than the classical transforms.
\end{abstract}

\section{Introduction}\label{sec:intro}
The goal of a moment transformation is to compute statistical moments of random variables transformed through an arbitrary nonlinear function.
The moment transforms find their uses in such problems as sensor system design \citep{Zangl2008}, optimal control \citep{Heine2006,Ross2015} and they are also an indispensable part of local nonlinear filters and smoothers \citep{Sandblom2012,Dunik2013,Saerkkae2016,Wu2006,Tronarp2016}. 
These algorithms estimate a state of dynamical systems based on noisy measurements and are applied in solving a broad array of engineering problems such as aircraft guidance \citep{Smith1962}, GPS navigation \citep{Grewal2007}, weather forecasting \citep{Gillijns2006}, telecommunications \citep{Jiang2003} and finance \citep{Bhar2010} to name a few.

Recursive nonlinear filters can be divided into two categories: global and local.
Particle filters are typical representatives of the global filters, which are characterized by weaker assumptions and a higher computational demand. 
On the contrary, the local filters trade off more limiting assumptions for computational simplicity.
For tractability reasons, the local filters often leverage the joint Gaussianity assumption of state and measurement.
The main problem then lies in computation of transformed means and covariances, which are subsequently combined with a measurement in an update rule to produce the filtering state estimate.
Examples of well-known local nonlinear filters include the unscented Kalman filter (UKF) \citep{Julier2000}, the cubature Kalman filter (CKF) \citep{Arasaratnam2009} and the Gauss-Hermite Kalman filter (GHKF) \citep{Wu2006}, which are collectively known as \emph{sigma-point} filters, and which are characterized by their reliance on the classical numerical integration schemes.

A limitation of classical integral approximations, such as the Gauss-Hermite quadrature, is that they are specifically designed to achieve zero error on a narrow class of functions (typically polynomials up to a given degree). 
Since many nonlinearities encountered in practical problems do not fall into this category, the integrals are, more often than not, approximated with errors which go unaccounted for.
Even though error estimates for the classical rules exist, their computation is tedious in practice as they require higher order derivatives \citep{Gautschi2004}.

In recent years, the Bayesian quadrature (BQ) has been gaining attention as an exciting alternative for approximate evaluation of integrals. 
According to \citet{Diaconis1988}, the origins of this method lie as far as Poincaré's publication \citep{Poincare1896} from 1896. 
Later developments came with O'Hagan's work on Bayes-Hermite quadrature \citet{OHagan1991} and the Bayesian Monte Carlo \citep{Rasmussen2003a}.
Seeing BQ as an instance of a probabilistic approach to numerical computing spawned an emerging field of probabilistic numerics with its share of contributions \citep{Osborne2012,Briol2015,Oates2015}.
Contrary to the classical rules, the BQ minimizes an average error on a wider class of functions \citep{Minka2000,Briol2015}.
The numerical integration process is treated as a problem of Bayesian inference, where, in accordance with the Bayesian paradigm, an integral prior is transformed into a posterior by conditioning on obtained evaluations of the integrated function.
The result of integration is much more informative, because it is no longer a single value, but an entire distribution.
The posterior mean estimates the integral value, whereas the posterior variance can be construed as a model of the integration error. 
The BQ approach is uniquely suited to the purpose of our article, which is accounting for the integration errors in the moment transform process. 

Application of the BQ in nonlinear sigma-point filtering was first investigated by \citet{Sarkka2014}, where the authors elucidate connections between the BQ and the classical rules, but their algorithms do not make use of the integral variance.
When our GPQ moment transform is applied in the nonlinear sigma-point filtering, it leads to a similar approach to the previously proposed Gaussian process assumed density filter (GP-ADF) \citep{Deisenroth2009} and the Gaussian process unscented Kalman filter (GP-UKF) \citep{Ko2007}.
Both GP-ADF and GP-UKF use the GP models for system identification that takes place prior to running the filters. 
In our case, however, the crucial difference lies in the fact that the resulting GPQ filters do not require a system identification phase.

In our previous contribution \citep{Prueher2016}, we showed how to incorporate integral variance into a nonlinear sigma-point filtering algorithm.
This article crystallizes the results of the previous publication into a widely applicable general GPQ moment transform. 
Applications in sigma-point filtering are enriched with a target tracking example and an additional numerical experiments on common nonlinear coordinate transformation are provided.
We also further analyse properties of the proposed transform and give a theoretical proof for positive semi-definiteness of the resulting covariance matrix.

The rest of this article is organized as follows.
\Cref{sec:problem} introduces the general problem of the classical quadrature based moment transforms. 
\Cref{sec:gpq_transform} outlines the Gaussian process quadrature (GPQ). 
In \Cref{sec:moment_transform}, we describe the proposed moment transform based on the GPQ. 
Numerical experiments and performance evaluations are in \Cref{sec:experiments}.
Finally, \Cref{sec:conclusion} concludes the article.

\section{Problem Statement}\label{sec:problem}
Consider an input Gaussian random variable \( \inVar \in \R^\inDim \) with mean \( \inMean \) and covariance \( \inCov \) transformed through a nonlinear vector function
\begin{equation}\label{eq:nonlinearity}
	\outVar = \nlf(\inVar) \qquad \N*[\inVar][\inMean][\inCov]
\end{equation}
producing an output \( \outVar \in \R^\outDim \).
Even though the output is not Gaussian (due to the nonlinearity of \( \nlf \)), we approximate both variables as jointly Gaussian distributed, so that
\begin{equation}\label{eq:gaussian_joint_xy}
	\N*[\bmqty{\inVar \\ \outVar}][\bmqty{\inMean \\ \outMean}][\bmqty{\inCov & \inoutCov \\ \inoutCov\T & \outCov}],
\end{equation}
where the transformed moments are given by the following Gaussian weighted integrals
\begin{align}
	\outMean 	= \E[\inVar]{\nlf(\inVar)} &= \integral \nlf(\inVar)\N[\inVar][\inMean][\inCov] \d[\inVar], \label{eq:outMean_integral} \\
	\outCov  	= \Cov[\inVar]{\nlf(\inVar)} &= \integral (\nlf(\inVar) - \outMean)(\nlf(\inVar) - \outMean)\T \N[\inVar][\inMean][\inCov] \d[\inVar], \label{eq:outCov_integral} \\
	\inoutCov 	= \Cov[\inVar]{\inVar, \nlf(\inVar)} &= \integral (\inVar - \inMean)(\nlf(\inVar) - \outMean)\T \N[\inVar][\inMean][\inCov] \d[\inVar], \label{eq:inoutCov_integral}
\end{align}
where \( \N[\inVar][\inMean][\inCov] \) denotes probability density of Gaussian random variable.
The goal of a moment transform is to compute the moments in \crefrange{eq:outMean_integral}{eq:inoutCov_integral} given the moments of the input variable.
Since \( \nlf \) is nonlinear, the integrals cannot be solved analytically in general and have to be approximated by a numerical quadrature.
An integral w.r.t. a Gaussian with arbitrary mean and covariance can be converted to an integral over a standard Gaussian, such that
\begin{equation}\label{eq:gaussian_integral_decoupling}
	\integral \nlf(\inVar)\N[\inVar][\inMean][\inCov] \d[\inVar] = \integral \nlf(\inMean + \inCovFct\inVarUnt)\N[\inVarUnt][\zeroVec][\eye] \d[\inVarUnt],
\end{equation}
where we used a change of variables \( \inVar = \inMean + \inCovFct\inVarUnt \) with \( \inCov = \inCovFct\inCovFct\T \). 
The standard numerical integration rules can now be applied, which leads to a weighted sum approximation
\begin{equation}\label{eq:gaussian_integral_quadrature}
	\integral \nlf(\inMean + \inCovFct\inVarUnt)\N[\inVarUnt][\zeroVec][\eye] \d[\inVarUnt] 
	= \sum\limits_{i=1}^{\trNum} w_i\nlf(\inMean + \inCovFct\inVarUnt_i),
\end{equation}
where \( w_i \) are the quadrature weights and \( \inVar_i = \inMean + \inCovFct\inVarUnt_i \) are the sigma-points (evaluation points, design points, abscissas), both of which are prescribed by the specific quadrature rule to satisfy certain optimality criteria.
The vectors \( \inVarUnt_i \) are called unit sigma-points.


For instance, the $r$-th order Gauss-Hermite (GH) rule \citep{Gautschi2004,Ito2000} uses sigma-points, which are determined as the roots of the $r$-th degree univariate Hermite polynomial $H_r(x)$. 
Integration of vector valued functions ($ \inDim > 1 $) is handled by a multidimensional grid of points formed by the Cartesian product, leading to the exponential growth ($N = r^\inDim$) w.r.t. dimension \( \inDim \).
The GH weights are computed according to \citep{Saerkkae2013} as
\begin{equation}
	w_i = \frac{r!}{[rH_{r-1}(x^{(i)})]^2} \enspace .
\end{equation}
The rule incurs no integration error if the integrand is a (multivariate) polynomial of \emph{pseudo-degree} \( \leq 2r - 1 \).
The well-known Unscented transform (UT) \citep{Julier2000} is also a simple quadrature rule, that uses $ \trNum = 2\inDim + 1 $ deterministically chosen sigma-points \( \inVar_i = \inMean + \inCovFct\inVarUnt_i \) with unit sigma-points defined as columns of the matrix
\begin{equation}
	\bmqty{\inVarUnt_0 & \inVarUnt_1 & \ldots & \inVarUnt_{2\inDim}} = \bmqty{\mathbf{0} & c\eye_\inDim & -c\eye_\inDim}
\end{equation}
where $ \mathbf{I}_\inDim $ denotes $ \inDim\times\inDim $ identity matrix. 
The corresponding UT weights are defined by
\begin{equation}
	w_0 = \frac{\kappa}{\inDim+\kappa}, \quad w_i = \frac{1}{2(\inDim+\kappa)}, \quad i = 1, \ldots, 2\inDim
\end{equation}
with scaling factor $ c = \sqrt{\inDim+\kappa} $.
The UT rule can integrate (multivariate) polynomials of \emph{total degree} \( \leq 3 \) without incurring approximation error.
Spherical-radial (SR) integration rule, which is a basis of the CKF \citep{Arasaratnam2009}, is very similar to the UT, but lacks the centre point.
Thus it uses \( \trNum = 2\inDim \) sigma-points given by 
\begin{equation}
\bmqty{\inVarUnt_1 & \ldots & \inVarUnt_{2\inDim}} = \bmqty{c\eye_\inDim & -c\eye_\inDim}
\end{equation}
with \( c = \sqrt{\inDim} \) and weights \( w_i = 1/2\inDim, \quad i = 1, \ldots, 2\inDim \).
The UT and SR rule are all instances of the fully symmetric rules \citep{McNamee1967}.
Together with GH all of these are examples of classical numerical quadratures.

In the next section, we introduce an alternative view of numerical integration, upon which, we base our proposed moment transform.

\section{Gaussian Process Quadrature}\label{sec:gpq_transform}
The main difference between the classical quadrature and the Bayesian quadrature, is that in the Bayesian case the whole numerical procedure is viewed as a probabilistic inference, where, after conditioning on the data, an entire density over the solution is obtained as a result (rather than a single value).
That is to say, we start with a prior distribution over the integral which is then transformed into a posterior distribution given our sigma-points and function evaluations.
Prior over the integral is induced by putting a prior on the integrated function itself.

\subsection{Gaussian Process Regression Model}\label{ssec:gpq_gp_regression}
Uncertainty over functions is naturally expressed by a stochastic process. 
In Bayesian quadrature, Gaussian processes (GP) are used for their favourable analytical properties. 
Gaussian process is a collection of random variables indexed by elements of an index set, any finite number of which has a joint Gaussian density \citep{Rasmussen2006}. 
That is, for any finite set of indices \( \mb{x}' = \bmqty{\inVar'_1 & \ldots & \inVar'_m} \), it holds that
\begin{equation}
	\N*[\bmqty{ \nlf*(\inVar'_1) & \ldots & \nlf*(\inVar'_m) }\T][\zeroVec][\kerMat]
\end{equation}
where the kernel (covariance) matrix \( \kerMat \) is made up of pair-wise evaluations of the kernel (covariance) function. 
The element of \( \kerMat \) at position \( (i, j) \) is given by \( \bqty{\kerMat}_{ij} \,=\, \kerf(\inVar_i, \inVar_j) = \Cov[\nlf*]{\nlf*(\inVar_i),\, \nlf*(\inVar_j)} \), where \( \mathbb{C} \) is the covariance operator. 
Choosing a kernel, which in principle can be any symmetric positive definite function of two arguments, introduces modelling assumptions about the underlying function.
Since the GP regression is a non-parametric model, it is more expressive than a parametric fixed-order polynomial regression models used in construction of the classical quadrature rules.
In Bayesian terms, choosing a kernel specifies a GP prior \( p(\nlf*) \) over functions.
Updating the prior with the data \( \D \,=\, \left\lbrace \pqty{\mathbf{x}_i,\ \nlf*(\mathbf{x}_i)} \right\rbrace_{i=1}^\trNum \), which consist of the sigma-points \( \mb{x} \,=\, \bmqty{\inVar_1 & \ldots & \inVar_\trNum} \) and the function evaluations \( \gpObs \,=\, \bmqty{\nlf*(\inVar_1) & \ldots & \nlf*(\inVar_\trNum)}\T \), leads to a GP posterior \( p(\nlf* \mid \D) \) with predictive mean and variance given by
\begin{align}
	\E[\nlf*]{\nlf*(\inVar) \mid \D} &= m_g(\inVar) = \vb{\kerf}\T(\inVar)\kerMat\I\gpObs, \label{eq:gpr_mean} \\
	\V[\nlf*]{\nlf*(\inVar) \mid \D} &= \sigma^2_{\nlf*}(\inVar) = \kerf(\inVar, \inVar) - \vb{\kerf}\T(\inVar)\kerMat\I\vb{\kerf}(\inVar), \label{eq:gpr_variance}
\end{align}
where the \( i \)-th element of \( \vb{\kerf}(\inVar) \) is \( \bqty{\vb{\kerf}(\inVar)}_i = \kerf(\inVar, \inVar_i) \) and \( \mathbb{V} \) denotes the variance operator.
These simple equations follow from the formula for conditional Gaussian densities \citep{Rasmussen2006}.
\Cref{fig:gp_regression} depicts predictive moments of the GP posterior. 
Notice, that in places where the function evaluations are lacking, the GP model is more uncertain.
\begin{figure}[h]
	\centering
	\input{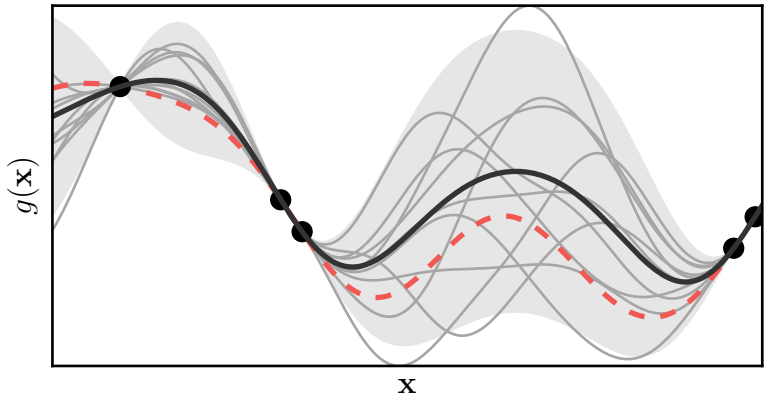}
	\caption[]{True function (dashed), GP posterior mean (solid), observed function values (dots) and GP posterior samples (grey). The shaded area represents GP posterior predictive uncertainty ($ \pm 2\, \sigma_{\nlf*}(\inVar) $), which is collapsed near the observations.}
	\label{fig:gp_regression}
\end{figure}

At first, introducing randomness\footnote{Here we mean \emph{epistemic uncertainty}, which is due to the lack of knowledge not due to inherent randomness.} over a \emph{known}\footnote{In a sense, that it can be evaluated for any argument.} integrand may seem quite counter-intuitive. 
However, consider the fact that the quadrature rule of the form \labelcref{eq:gaussian_integral_quadrature} only sees the integrand through a limited number of function values, which effectively means, that the rule is unaware of the function behaviour in areas where no evaluations are available.
The GP regression model then serves as an instrument allowing us to acknowledge this reality.

\subsection{Integral Moments}
Using a GP for modelling the integrand has a great analytical advantages. 
Note that, since we use a GP, the posterior density over the integral is Gaussian, which is due to the fact that an integral is a linear operator acting on a Gaussian distributed function.\footnote{This is a generalization of the invariance property of Gaussians under affine transformations.}
The posterior mean and variance of the integral \( \E[\inVar]{\nlf*(\inVar)} = \integral \nlf*(\inVar)p(\inVar) \d[\inVar] \) are \citep{Rasmussen2003a}
\begin{align}
	\E[\nlf*]{\E[\inVar]{\nlf*(\inVar)} \mid \D} 
	&= \E[\inVar]{\E[\nlf*]{\nlf*(\inVar) \mid \D}} = \E[\inVar]\big{\vb{\kerf}\T(\inVar)}\kerMat\I\gpObs, \label{eq:integral_mean} \\
	\V[\nlf*]{\E[\inVar]{\nlf*(\inVar)} \mid \D} 
	&= \E[\inVar,\inVar']{\Cov[\nlf*]{\nlf*(\inVar), \nlf*(\inVar')\mid \D}} \nonumber \\ 
	&= \E[\inVar,\inVar']{\kerf(\inVar, \inVar')} - \E[\inVar]\big{\vb{\kerf}\T(\inVar)}\kerMat\I\E[\inVar']{\vb{\kerf}(\inVar')}. \label{eq:integral_variance}
\end{align}
From \cref{eq:integral_mean}, we can see that the mean of the integral is identical to integrating the GP posterior mean, which effectively serves as an approximation of the integrated function.
The posterior variance of the integral, given by \cref{eq:integral_variance}, can be construed as a model of the integration error.
The \Cref{fig:bayesian_quadrature} depicts the schematic view of the GP quadrature, where integral is applied to a GP distributed integrand yielding a Gaussian density over the value of the integral.
Alternatively, one could imagine integrating each realization of the GP posterior separately, yielding a different result every time.

\begin{figure}[!h]
	\centering
	\input{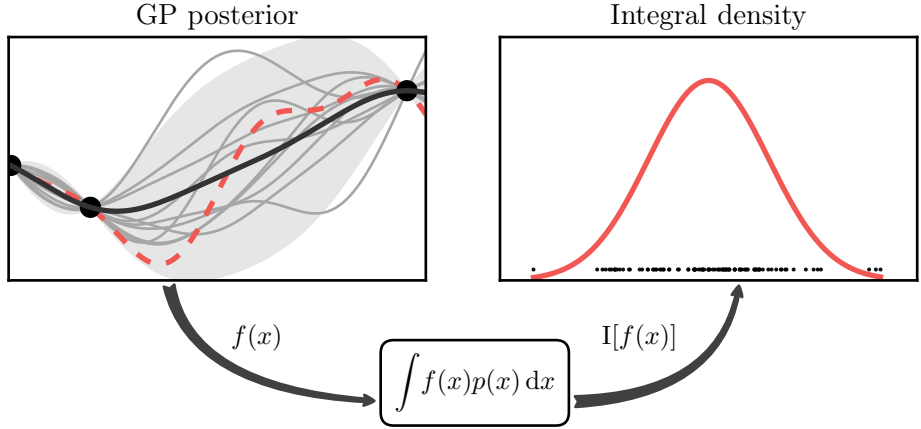}
	\caption{Gaussian process quadrature. GP distributed function is mapped through a linear operator to yield a Gaussian density over its solution. The GP posterior mean approximation (solid black) of the true function (dashed) and the GP predictive variance (gray band) are based on the function evaluations (dots).
		}
	\label{fig:bayesian_quadrature}
\end{figure}
In light of this view, we could think of the classical quadratures as returning a Dirac distribution over the solution, where all the probability mass is concentrated at one value.
In this regard, Bayesian quadrature rule is more realistic, because it is able to acknowledge uncertainty in the integrand due to the minimal number of available evaluations.

Also worth noting, is that the classical quadrature rules define precise locations of sigma-points, whereas the BQ does not prescribe any point sets, which raises questions about their placement.
In \citep{OHagan1991,Minka2000}, the optimal point set is determined by minimizing the posterior integral variance \eqref{eq:integral_variance}.
Another approach developed in \citep{Osborne2012a} uses a sequential active sampling scheme.
Applications of GPQ in this article rely on point sets of the classical rules.


\section{GPQ Moment Transformation}\label{sec:moment_transform}
In this section, we propose the moment transform based on the GPQ and show how it implicitly utilizes posterior integral variance in the moment transformation process.
First, we define a general GPQ moment transform, which is applicable for any kernel function, and then give relations for a concrete transform based on the popular RBF (Gaussian) kernel.

We begin by employing a GPQ for approximate evaluation of the moment integrals in \crefrange{eq:outMean_integral}{eq:inoutCov_integral}. 
For a moment, consider a case when the nonlinearity in \cref{eq:nonlinearity} is a scalar function \( \nlf*(\inVar):\ \R^\inDim\to\R \). 
Since the source of uncertainty is now, not only in the input \( \inVar \), but the nonlinearity \( \nlf* \) as well, the transformed moments also need to reflect this fact. 
The GPQ transform then approximates the moments as follows
\begin{align}\label{eq:gpq_mt_approximation}
	\E{\outVar*} = \E[\inVar]{\nlf*(\inVar)} &\approx \E[\nlf*, \inVar]{\nlf*(\inVar)} \\
	\V{\outVar*} = \V[\inVar]{\nlf*(\inVar)} &\approx \V[\nlf*, \inVar]{\nlf*(\inVar)} \\
	\Cov{\inVar, \outVar*} = \Cov[\inVar]{\inVar, \nlf*(\inVar)} &\approx \Cov[\nlf*, \inVar]{\inVar, \nlf*(\inVar)}
\end{align}
where, using the law of total expectation and variance, we can further write
\begin{align}
	\E[\nlf*, \inVar]{\nlf*(\inVar)} 
	&= \E[\nlf*]{\E[\inVar]{\nlf*(\inVar)}} = \E[\inVar]{\E[\nlf*]{\nlf*(\inVar)}}, \label{eq:gpq_mt_mean}\\
	\V[\nlf*, \inVar]{\nlf*(\inVar)} 
	&= \E[\nlf*]{\V[\inVar]{\nlf*(\inVar)}} + \V[\nlf*]{\E[\inVar]{\nlf*(\inVar)}}  \label{eq:gpq_mt_variance_0}\\
	&= \E[\inVar]{\V[\nlf*]{\nlf*(\inVar)}} + \V[\inVar]{\E[\nlf*]{\nlf*(\inVar)}}, \label{eq:gpq_mt_variance_1}\\
	\Cov[\nlf*, \inVar]{\inVar, \nlf*(\inVar)} 
	&= \E[\inVar]{\inVar\E[\nlf*]{\nlf*(\inVar)}} - \E[\inVar]{\inVar}\E[\nlf*, \inVar]{\nlf*(\inVar)}. \label{eq:gpq_mt_crosscovariance}
\end{align}
The \cref{eq:gpq_mt_mean} shows that the mean of the integral is equivalent to integrating the GP mean function. 
Since the variance decompositions in \crefrange{eq:gpq_mt_variance_0}{eq:gpq_mt_variance_1} are equivalent, both can be used to achieve the same goal. 
The form \labelcref{eq:gpq_mt_variance_1} was utilized in derivation of the GP-ADF \citep{Deisenroth2012}, which relies on the solution to the problem of prediction with GPs at uncertain inputs \citep{Girard2003}. 
So, even though these results were derived to solve a seemingly different problem, we point out, that by using the form \labelcref{eq:gpq_mt_variance_1}, the uncertainty of the mean integral (as seen in the last term of \cref{eq:gpq_mt_variance_0}) is implicitly reflected in the resulting covariance.
Furthermore, the form \labelcref{eq:gpq_mt_variance_1} is preferable, because it is more amenable to analytical expression and implementation.
Note, that for the deterministic case, when the \emph{integrand} variance \( \V[\nlf*]{\nlf*(\inVar)} = 0 \) and the \emph{integral} variance \( \V[\nlf*]{\E[\inVar]{\nlf*(\inVar)}} = 0 \), the \crefrange{eq:gpq_mt_mean}{eq:gpq_mt_crosscovariance} fall back to the classical expressions given by \crefrange{eq:outMean_integral}{eq:inoutCov_integral}.
Compared to the deterministic case the transformed GPQ variance is inflated by the uncertainty in \( \nlf* \).

\subsection{Derivations of the transformed moments}
So far we have stated the results only for the case of scalar function.
In the following summary, general vector functions \( \nlf(\inVar): \R^\inDim\to\R^\outDim \) are modelled by a single GP. 
That is, one GP models every output dimension using the same values of kernel parameters.
Note, that in the following derivations we omit the conditioning on data in the GP predictive moments and use the shorthand \( \E[\nlf]{\nlf(\inVar)} \triangleq \E[\nlf]{\nlf(\inVar) \mid \D},\ \Cov[\nlf]{\nlf(\inVar)} \triangleq \Cov[\nlf]{\nlf(\inVar) \mid \D} \).

Expressions for GPQ transformed moments are derived by plugging in the GP predictive moments from \crefrange{eq:gpr_mean}{eq:gpr_variance} into the general expressions in \crefrange{eq:gpq_mt_mean}{eq:gpq_mt_crosscovariance}.
The transformed mean in \cref{eq:gpq_mt_mean} thus becomes
\begin{equation}\label{eq:gpq_mt_mean_derivation}
	\outMeanApp = \E[\inVar]{\E[\nlf]{\nlf(\inVar)}} = \mb{y}\T\kerMat\I\E[\inVar]\big{\vb{\kerf}(\inVar)} = \mb{y}\T\wm.
\end{equation}
The transformed covariance in \cref{eq:gpq_mt_variance_1} can be decomposed
\begin{align}\label{eq:gpq_mt_covariance_derivation_0}
	\outCovApp = \Cov[\nlf, \inVar]{\nlf(\inVar)} 
	&= \Cov[\inVar]{\E[\nlf]{\nlf(\inVar)}} + \E[\inVar]{\Cov[\nlf]{\nlf(\inVar)}} \\
	&= \E[\inVar]\big{\E[\nlf]{\nlf(\inVar)}\E[\nlf]{\nlf(\inVar)}\T} - \outMeanApp\outMeanApp\T + \E[\inVar]{\Cov[\nlf]{\nlf(\inVar)}} \\
	&= \mb{y}\T\kerMat\I\E[\inVar]\big{\vb{\kerf}(\inVar)\vb{\kerf}\T(\inVar)}\kerMat\I\mb{y} - \outMeanApp\outMeanApp\T + \gpExpVar\eye,
\end{align}
where \( \gpExpVar = \E[\inVar]\big{\sigma^2_{\nlf*}(\inVar)} = \E[\inVar]{\kerf(\inVar, \inVar)} - \trace\pqty\big{\E[\inVar]\big{\vb{\kerf}(\inVar)\vb{\kerf}\T(\inVar)}\kerMat\I} \).
The diagonal matrix in the last term reflects the fact that the outputs of \( \nlf \) are not correlated (modelled independently).
Finally, the cross-covariance becomes
\begin{align}\label{eq:gpq_mt_crosscovariance_derivation}
\inoutCovApp = \Cov[\nlf, \inVar]{\inVar, \nlf(\inVar)} 
&= \E[\inVar]\big{\inVar\E[\nlf]{\nlf(\inVar)}\T} - \E[\inVar]{\inVar}\E[\inVar]{\E[\nlf]{\nlf(\inVar)}} \\
&= \E[\inVar]\big{\inVar\vb{\kerf}\T(\inVar)}\kerMat\I\mb{y} - \inMean\outMeanApp\T
\end{align}
Summary of the proposed GPQ moment transform is given below.
\ \\
\ \\\textbf{General GPQ moment transform}\\
The general GPQ based Gaussian approximation to the joint distribution of $ \inVar $ and a transformed random variable $ \outVar = \nlf(\inVar) $, where $ \N*[\inVar][\inMean][\inCov] $, is given by
\begin{equation}
	\N*[\bmqty{\inVar \\ \outVar}][\bmqty{\inMean \\ \outMeanApp}][\bmqty{\inCov & \inoutCovApp \\ \inoutCovApp\T & \outCovApp}]
\end{equation}
where the transformed moments are computed as
\begin{align}
	\outMeanApp 	&= \mb{y}\T\wm, \label{eq:gpq_mean_out} \\
	\outCovApp 		&= \mb{y}\T\wc\mb{y} - \outMean\outMean\T + \gpExpVar\eye, \label{eq:gpq_cov_out} \\
	\inoutCovApp	&= \wcc\mb{y} - \inMean\outMean\T, \label{eq:gpq_covio_out} \\
	\gpExpVar 		&= \bar{k} - \trace\pqty\big{\kerCov\kerMat\I}, \label{eq:gpq_expected_gp_variance}
\end{align}
and where $ \bqty\big{\mb{y}}_{*e} = \bmqty{y^e_1 & \ldots & y^e_\trNum}\T $ are the function values of the $ e $-th output dimension of $ \nlf(\inVar) $.
The kernel matrix $ \kerMat $ is defined in \cref{ssec:gpq_gp_regression} and the GPQ weights are 
\begin{equation}\label{eq:gpq_weights}
	\wm = \kerMat\I\kerMean,\ \wc = \kerMat\I\kerCov\kerMat\I \ \textnormal{and}\ \wcc = \kerCCov\kerMat\I, 
\end{equation}
where
\begin{align}
	\bqty{\kerMean}_{i} 	&= \E[\inVar]{\kerf\pqty{\inVar, \inVar_i}}, \label{eq:gpq_kernel_mean} \\
	\bqty{\kerCov}_{ij} 	&= \E[\inVar]{\kerf\pqty{\inVar, \inVar_i}\kerf\pqty{\inVar, \inVar_j}}, \label{eq:gpq_kernel_covariance} \\
	\bqty{\kerCCov}_{*j} 	&= \E[\inVar]{\inVar\kerf\pqty{\inVar, \inVar_j}}, \label{eq:gpq_kernel_crosscovariance} \\
	\bar{k}					&= \E[\inVar]{\kerf\pqty{\inVar, \inVar}}.
\end{align}
The sigma-points \( \inVar_i \) can be chosen \emph{arbitrarily}.

\subsection{Properties of the general GPQ transform}
An important requirement of moment transforms is that they produce valid covariance matrices. 
\Cref{thm:gpq_psd} given below states that the proposed GPQ transform always produces positive semi-definite covariance matrix.
For proof we use the following lemma.
\begin{lem}\label{lem:horn}
	For any \( m\times n \) matrix \( \mb{X} \) and positive definite \( n\times n \) matrix \( \mb{A} \), the matrix \( \mb{XAX}\T \) is positive semi-definite.
\end{lem}
\begin{proof}
	See \citep[Observation 7.1.6, p. 399]{Horn1990}.
\end{proof}
In the following, let \( \mb{a} \succeq 0\ \Leftrightarrow\ \vb{x}\T\mb{a}\vb{x} \geq 0,\ \forall \vb{x}\in\R^n \) for any \( n\times n \) matrix \( \mb{a} \).
\begin{thm}[]\label{thm:gpq_psd}
	The GPQ transformed covariance is positive semi-definite.
\end{thm}
\begin{proof}
	Using the expressions for the GPQ weights from \cref{eq:gpq_weights}, we can write 
	\[
	\outCov = \mb{y}\T\kerMat\I\pqty\big{\kerCov - \kerMean\kerMean\T}\kerMat\I\mb{y} + \gpExpVar\eye = \mb{z}\T\widetilde{\kerCov}\mb{z} + \gpExpVar\eye = \widetilde{\outCov} + \gpExpVar\eye, 
	\]
	where \( \widetilde{\outCov} = \mb{z}\T\widetilde{\kerCov}\mb{z},\ \mb{z} = \kerMat\I\mb{y} \) and \( \widetilde{\kerCov} = \kerCov - \kerMean\kerMean\T \).
	We recognize that \( \widetilde{\kerCov} = \Cov{\vb{\kerf}(\inVar)} = \E{\vb{\kerf}(\inVar)\vb{\kerf}\T(\inVar)} - \E{\vb{\kerf}(\inVar)}\E{\vb{\kerf}(\inVar)}\T \). 
	From the property of covariance matrices it follows that \( \widetilde{\kerCov} \succeq 0 \).
	The \Cref{lem:horn} implies that \( \widetilde{\outCov} = \mb{z}\T\widetilde{\kerCov}\mb{z} \succeq 0 \) for any matrix \( \mb{Z} \).
	Finally, since \( \gpExpVar \geq 0 \), we have that \( \outCov = \widetilde{\outCov} + \gpExpVar\eye \succeq 0 \).
\end{proof}
Evidently, the GPQ moment transform hinges upon the kernel expectations given by \crefrange{eq:gpq_kernel_mean}{eq:gpq_kernel_crosscovariance}. 
Since we are already using one quadrature to approximate moments, it is thus preferable that these expectations be analytically tractable.
A list of tractable kernel-density pairs is provided in \citep{Briol2015}.
A popular choice in many applications is an RBF (Gaussian) kernel, expectations of which are summarized below.
\begin{thm}[GPQ transform with RBF kernel]
	Assuming a change of variables has taken place in the Gaussian weighted integrals given by \cref{eq:gaussian_integral_decoupling} and the kernel is of the form
	\begin{equation}\label{eq:kernel_rbf}
		\kerf\pqty{\inVarUnt, \inVarUnt'} = \rbfScale^2 \exp\pqty\big{-\tfrac{1}{2}\pqty{\inVarUnt - \inVarUnt'}\T\rbfLam\I\pqty{\inVarUnt - \inVarUnt'}},
	\end{equation}
	where \( \rbfScale \) is scaling parameter and \( \rbfLam = \diag{\bmqty{\ell^2_1 & \ldots & \ell^2_\inDim}} \) is lengthscale, then the expectations given by \crefrange{eq:gpq_kernel_mean}{eq:gpq_kernel_crosscovariance} take on the form
	\begin{align}
		\bqty{\kerMean}_i		&= \alpha^2\vqty{\rbfLam\I + \eye}^{-\tfrac{1}{2}} \exp\pqty\big{-\tfrac{1}{2}\inVarUnt_i\T\pqty{\rbfLam + \eye}\I\inVarUnt_i}, \label{eq:gpq_rbf_mean} \\
		\bqty{\kerCov}_{ij} 	&= \alpha^4\vqty{2\rbfLam\I + \eye}^{-\tfrac{1}{2}} 
		\exp\pqty{ -\tfrac{1}{2}\pqty{\inVarUnt_i\T\rbfLam\I\inVarUnt_i + \inVarUnt_j\rbfLam\I\inVarUnt_j - \vb{z}_{ij}\T\pqty{2\rbfLam\I + \eye}\I\vb{z}_{ij}} },\label{eq:gpq_rbf_covariance} \\
		\bqty{\kerCCov}_{*j} 	&= \alpha^2\vqty{\rbfLam\I + \eye}^{-\tfrac{1}{2}} \exp\pqty\big{-\tfrac{1}{2}\inVarUnt_j\T\pqty{\rbfLam + \eye}\I\inVarUnt_j} \pqty{\rbfLam + \eye}\I\inVarUnt_j, \label{eq:gpq_rbf_crosscovariance} \\
		\bar{k}					&= \alpha^2
	\end{align}
	where \( \vb{z}_{ij} = \rbfLam\I(\inVarUnt_i + \inVarUnt_j) \).
\end{thm}
\begin{proof}
	Expressions can be derived by writing the RBF kernel as a Gaussian and making use of the formulas for the product of two Gaussian densities (and the normalizing constant).
	For space reasons, we omit the fairly straightforward, but nevertheless lengthy and tedious, derivations.
\end{proof}

It is now evident that the GPQ transformed moments depend on the kernel parameters which need to be set prior to computing the weights.
The form of \( \rbfLam \) in the RBF kernel formulation above exhibits, so called, automatic relevance determination (ARD).
That is to say, by optimizing the lengthscales \( \ell_d \) dimensions contributing most to the variability in the data can be discovered, where a small \( \ell_d \) would indicate high relevance of the \( d \)-th dimension.
A typical approach in GP regression would be to optimize the kernel parameters by marginal likelihood (evidence) maximization.
However, in the BQ setting this method would likely yield unreliable parameter estimates due to the inherently minimal amount of data available.
For these reasons, we resorted to a manual choice of the parameter values which were mostly informed by the prior knowledge of the integrated function. 
In the following theorem, we prove the weight independence on the kernel scaling parameter.

\begin{thm}[Kernel scaling independence]
	Assume a scaled version of a kernel is used, so that \( \bar{\kerf}(\inVar, \inVar') = c\cdot \kerf(\inVar, \inVar') \), then the weights of the GPQ transform given in \cref{eq:gpq_weights} are independent of the scaling parameter \( c \).
\end{thm}
\begin{proof}
	Define a scaled kernel matrix \( {\kerMat}' = c\kerMat,\) and scaled kernel expectations \(\ \bqty{{\kerMean}'}_i = c\E[\inVar]{\kerf\pqty{\inVar, \inVar_i}},\ \bqty{{\kerCov}'}_{ij} = c^2\E[\inVar]{\kerf\pqty{\inVar, \inVar_i}\kerf\pqty{\inVar, \inVar_j}},\ \bqty{{\kerCCov}'}_{*j} = c\E[\inVar]{\inVar\kerf\pqty{\inVar, \inVar_j}} \).
	Plugging into the expressions for the GPQ weights from the \cref{eq:gpq_weights}, we get
	\begin{align}
		\wm' 	&= {\kerMean}'\pqty{\kerMat'}\I = cc\I\kerMean\kerMat\I = \wm,  \\
		\wc' 	&= \pqty{\kerMat'}\I{\kerCov'}\pqty{\kerMat'}\I = c^2c^{-2}\kerMat\I\kerCov\kerMat\I = \wc, \\
		\wcc' 	&= {\kerCCov}'\pqty{\kerMat'}\I = cc\I\kerCCov\kerMat\I = \wcc.
	\end{align}
\end{proof}

\begin{cor}
	The kernel scaling affects only the additive term in the transformed covariance, which becomes \( \gpExpVar = c\bqty{\bar{k} - \trace\pqty\big{\kerCov\kerMat\I}}. \) The transformed mean and cross-covariance are unaffected by the scaling.
\end{cor}


\subsection{Moment transforms in filtering}
An important application area for the moment transforms is in local filtering algorithms, which estimate an evolving system state from noisy measurements.
Filters are essentially inference algorithms operating on a state-space model which is given by the two discrete-time equations
\begin{align}
	\stVar_{k} &= \dynf(\stVar_{k-1}) + \stNoise_k, \label{eq:ssm_dynamics}\\
	\obsVar_{k} &= \obsf(\stVar_k) + \obsNoise_k, \label{eq:ssm_observation_model}
\end{align}
describing the system dynamics and the state observation (measurement) process respectively.
The evolution of the system state \( \stVar_k \) is described by \cref{eq:ssm_dynamics} where \( \dynf: \R^\inDim\to\R^\inDim \) is the dynamics  and the \( \N*[\stNoise_k][\vb{0}][\stNoiseCov] \) is the state noise.
The measurements \( \obsVar_k \) are produced by mapping the state through the observation model \( \obsf: \R^\inDim\to\R^\outDim \) and adding the measurement noise \( \N*[\obsNoise_k][\vb{0}][\obsNoiseCov] \).
Typically, \( \outDim \leq \inDim \).
Let \( \obsVar_{1:k} = \Bqty{\obsVar_1, \ldots, \obsVar_k} \) denote a set of measurements up to time step \( k \). 
From a probabilistic standpoint, the filtering problem is about inferring a posterior distribution
\begin{equation}\label{eq:pdf_filtering_posterior}
	p(\stVar_k \mid \obsVar_{1:k}) \propto p(\stVar_k,\, \obsVar_{k} \mid \obsVar_{1:k-1}) = p(\obsVar_k \mid \stVar_{k})p(\stVar_k \mid \obsVar_{1:k-1})
\end{equation}
where the likelihood \( p(\obsVar_k \mid \stVar_{k}) \) is obtained from \cref{eq:ssm_observation_model} and the prior \( p(\stVar_k \mid \obsVar_{1:k-1}) \) is given by the Chapman-Kolmogorov equation
\begin{equation}\label{eq:pdf_chapman_kolmogorov}
	p(\stVar_k \mid \obsVar_{1:k-1}) = \integral p(\stVar_k \mid \stVar_{k-1})p(\stVar_{k-1} \mid \obsVar_{1:k-1}) \d[\stVar_{k-1}]
\end{equation}
where the transition density \( p(\stVar_k \mid \stVar_{k-1}) \) is obtained from the system dynamics given by \cref{eq:ssm_dynamics}.
Local Gaussian filters make the simplifying assumption that the state and measurement are jointly Gaussian distributed, so that 
\begin{equation}\label{eq:pdf_gassian_assumption}
	p(\stVar_k,\, \obsVar_{k} \mid \obsVar_{1:k-1}) \approx \N[\bmqty{\stVar_k \\ \obsVar_{k}}][\bmqty{\stMean_{k|k-1} \\ \obsMean_{k|k-1}}][\bmqty{\stCov_{k|k-1} & \stObsCov_{k|k-1} \\ \stObsCov_{k|k-1} & \obsCov_{k|k-1}}],
\end{equation}
where the index notation \( k|k-1 \) means that the relevant quantity at time \( k \) is computed from \( \obsVar_{1:k-1} \).
Advantage of this simplification is that the posterior is now parametrized by the conditional mean and covariance, which are available in closed form. 
Using an update rule the state and measurement moments of the joint in \cref{eq:pdf_gassian_assumption} 
\begin{align}
	\stMean_{k|k} &= \stMean_{k|k-1} + \stObsCov_{k|k-1}(\obsCov_{k|k-1})\I\pqty\big{\obsVar_k - \obsMean_{k|k-1}}, \label{eq:kalman_update_mean} \\
	\stCov_{k|k}  &= \stCov_{k|k-1} - \stObsCov_{k|k-1}(\obsCov_{k|k-1})\I(\stObsCov_{k|k-1})\T, \label{eq:kalman_update_covariance}
\end{align}
are combined to arrive at the approximate conditional mean and covariance of the state.

To compute state predictive moments \( \stMean_{k|k-1},\ \stCov_{k|k-1} \) a moment transform is applied in a setting where the input moments are \( \stMean_{k-1|k-1},\ \stCov_{k-1|k-1} \) and the nonlinearity is the dynamics \( \dynf(\stVar_{k-1}) \).
Similarly, the measurement moments \( \obsMean_{k|k-1},\ \obsCov_{k|k-1},\ \stObsCov_{k|k-1} \) are obtained by applying a moment transform on input moments \( \stMean_{k|k-1},\ \stCov_{k|k-1} \) with nonlinearity \( \obsf(\stVar_k) \).

\section{Experiments}\label{sec:experiments}
The proposed GPQ moment transform is first tested on a polar-to-Cartesian coordinate transformation while the later experiments focus on applications in nonlinear filtering.
In all cases the GPQ transform uses the RBF kernel given by \cref{eq:kernel_rbf}.
Since the sigma-point locations are not prescribed and their choice is entirely arbitrary, we used the point sets of the classical rules mentioned in \Cref{sec:problem} for all examples.
The acronym GPQKF denotes all nonlinear Kalman filters based on the GPQ regardless of which point set they use.

\subsection{Mapping from Polar to Cartesian Coordinates}

Conversion from polar to Cartesian coordinates is a ubiquitous nonlinearity appearing in radar sensors or laser range finders and is given by
\begin{equation}\label{eq:polar2cartesian}
	\bmqty{x \\ y} = 
	\bmqty{r\cos(\theta) \\ r\sin(\theta)}.
\end{equation}
Since the mapping is conditionally linear (for fixed \( \theta \)) and we use a kernel with ARD in the moment transform, we can exploit this fact and set the kernel lengthscales to \( \ell = \bmqty{60 & 6} \) while the scaling was set to \( \alpha = 1 \).
Note, that we set the lengthscale corresponding to range to a relatively large value. 
This is because the larger lengthscales in the kernel correspond to a slower variation in the approximated function.

We compared the performance of the spherical radial transform (SR), which is basis of the cubature Kalman filter \citep{Arasaratnam2009}, and the GPQ transform with SR points (GPQ-SR) for 100 different input moments.
The 10 different positions on a spiral in polar coordinates were chosen as input means \( \inMean_i = \bmqty{r_i & \theta_i} \).
For each mean we assigned 10 different input covariance matrices \( \inCov_j = \diag{\bmqty{\sigma^2_r & \sigma^2_{\theta,j}}} \), where \( \sigma_r = \SI{0.5}{m} \) and \( \sigma_{\theta,j} \in \bqty{\SI{6}{\degree},\, \SI{36}{\degree}} \) for \( j = 1, \ldots, 10 \).
\Cref{fig:polar_spiral} depicts the input means in polar coordinates.
\begin{figure}[!h]
	\centering
	\input{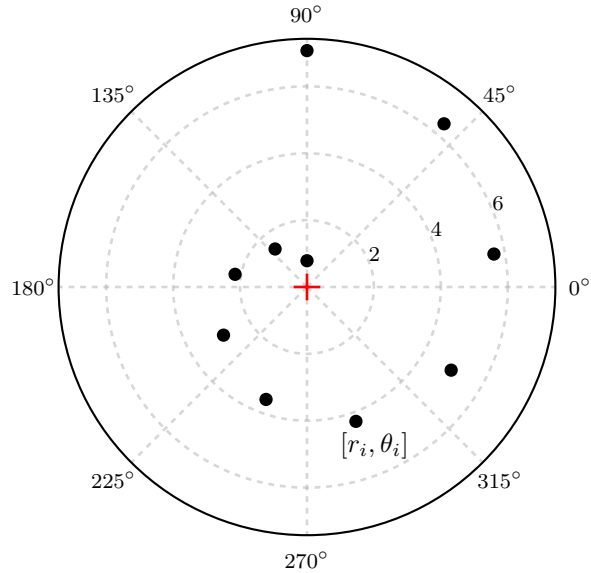}
	\caption{Input means are placed on a spiral. For each input mean \( \inMean_i = \bmqty{r_i & \theta_i} \) (black dot) the radius variance is fixed at \( \sigma_r = \SI{0.5}{m} \) and 10 azimuth variances are considered so that \( \sigma_\theta \in \bqty{\SI{6}{\degree},\, \SI{36}{\degree}} \).}
	\label{fig:polar_spiral}
\end{figure}
As a measure of an agreement between two Gaussian densities we used the symmetrized KL-divergence given by 
\begin{align}\label{eq:skl}
	\mathrm{SKL} 
	&= \frac{1}{2}\Bqty{\mathbb{KL}\pqty{\N[\outVar][\outMean][\outCov] \,\|\, \N[\outVar][\outMeanApp][\outCovApp]} + \mathbb{KL}\pqty{\N[\outVar][\outMeanApp][\outCovApp] \,\|\, \N[\outVar][\outMean][\outCov]}} \\
	&= \frac{1}{4}\Bqty{(\outMean - \outMeanApp)\T\outCov\I(\outMean-\outMeanApp) + (\outMeanApp - \outMean)\T\outCovApp\I(\outMeanApp - \outMean) + \trace(\outCov\I\outCovApp) + \trace(\outCovApp\I\outCov) - 2\outDim},
\end{align}
where \( E = \dim(\outVar) \).
The ground truth transformed mean \( \outMean \) and covariance \( \outCov \) were computed using the Monte Carlo method with \num{10000} samples.
Two SKL scores were considered; the average over means and an average over azimuth variances.

The \Cref{fig:polar2cartesian_skl} shows the SKL score calculated for each configuration on the spiral.
\begin{figure}[h!]
	\centering
	\input{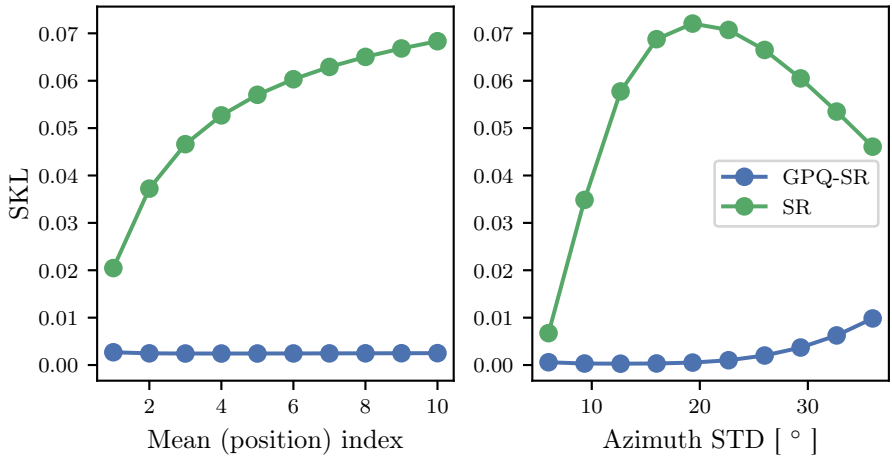}
	\caption{Performance comparison of the spherical radial (SR) and GPQ with SR points (GPQ-SR) moment transforms in terms of averaged symmetrized KL-divergence. Left: average over a range of azimuth variances; Right: average over the range of input means (positions on the spiral).}
	\label{fig:polar2cartesian_skl}
\end{figure}
The left pane of \Cref{fig:polar2cartesian_skl} shows results for individual means averaged over the azimuth variances, whereas the right pane displays averaged SKL over the means.
In both cases our proposed moment transform outperforms the classical quadrature transform with the same SR point set.

\subsection{UNGM}
The performance of nonlinear sigma-point filters based on GPQ transform was first tested in \citep{Prueher2016} on a univariate non-stationary growth model (UNGM), where the system dynamics and the observation model are given by
\begin{align}
	x_k &\,=\, \frac{1}{2}x_{k-1} \,+\, \frac{25x_{k-1}}{1+x^2_{k-1}} \,+\, 8\cos(1.2\,k) \,+\, q_{k-1} \, ,  \\
	z_k &\,=\, \frac{1}{20} x^2_{k-1} \,+\, r_k \, ,
\end{align}
with the state noise $q_{k-1} \sim \mathcal{N}(0, 10)$, measurement noise $r_k \!\sim\! \mathcal{N}(0, 1)$ and initial conditions $x_{0|0} \sim \mathcal{N}(0, 5)$.
This model is frequently used as a benchmark in nonlinear filtering \citep{Gordon1993,Kitagawa1996}.

For this problem, all of the considered GPQKFs used the same kernel scaling $\alpha = 1$. 
The lengthscale was set to $\ell = 3.0$ for the UT, $\ell = 0.3$ for SR and GH-5, and $\ell = 0.1$ for all higher-order GH sigma-point sets. 
The GPQKFs that used UT and GH sigma-points of order 5, 7, 10, 15 and 20 were compared with their classical quadrature-based counterparts, namely, the UKF and the GHKF of the same orders. 
The UKF operated with $\kappa = 0$.
We performed $100$ simulations, each for $ K=500 $ time steps. 

For evaluation of the filter performance, we used the root-mean-square error (RMSE)
\begin{equation}\label{eq:rmse}
	\mathrm{RMSE} = \sqrt{\frac{1}{K}\sum_{k=1}^{K} \pqty{\inVar_k - \vb{m}^\inVar_{k|k}}^2}
\end{equation}
to measure the overall difference between the state estimate $\vb{m}^\inVar_{k|k}$ and the true state \( \inVar_k \) across all time steps. 
The negative log-likelihood of the state estimate \( \vb{m}^\inVar_{k|k} \) and covariance \( \inCov^\inVar_k \)
\begin{equation}\label{eq:nll}
	\mathrm{NLL} = -\log p(\inVar_k \mid \vb{z}_{1:k}) = \frac{1}{2}\bqty{ \log \vqty{2\pi\inCov^\inVar_k} + (\inVar_{k} - \vb{m}^\inVar_{k|k})\T (\inCov^\inVar_k)\I (\inVar_{k} - \vb{m}^\inVar_{k|k}) }
\end{equation}
was used to measure the overall model fit \citep{Gelman2013}. 
As a metric that takes into account the estimated state covariance, the inclination indicator \citep{Li2006} given by
\begin{equation}\label{eq:nci}
	\nu = \frac{10}{K} \sum_{k=1}^{K} \log_{10}\frac{ \pqty\big{\inVar_k - \vb{m}^\inVar_{k|k}}^\top \pqty\big{\inCov^\inVar_{k|k}}\I \pqty\big{\inVar_k - \vb{m}^\inVar_{k|k}} }{ \pqty\big{\inVar_k - \vb{m}^\inVar_{k|k}}\T \mb{\Sigma}\I_{k} \pqty\big{\inVar_k - \vb{m}^\inVar_{k|k}} }
\end{equation}
was used, where $ \mb{\Sigma}_{k} $ is the sample mean-square-error matrix. 
The filter is said to be optimistic if it underestimates the actual error, which is indicated by $\nu~>~0$ and vice versa. 
A perfectly credible filter would provide $ \nu = 0 $, that is, it would neither overestimate nor underestimate the actual error.

\Crefrange{tab:rmse}{tab:nci} show average values of the performance criteria across simulations with bootstrapped estimates of $\pm 2$ standard deviations \citep{Wasserman2007}.
\begin{table} 
	\centering
	\smallskip
	\begin{tabular}{ l l >{$\,}l<{\:$} >{$\:}l<{\,$} }
		\toprule
		\text{Point set} & N & \text{GPQ} & \text{Classical} \\ 
		\midrule
		SR 		& 2 	& 6.157 \pm 0.071	& 13.652 \pm 0.253	\\
		UT 		& 3 	& 7.124 \pm 0.131 	&  7.103 \pm 0.130 	\\ 
		GH5 	& 5 	& 8.371 \pm 0.128 	& 10.466 \pm 0.198 	\\ 
		GH7 	& 7 	& 8.360 \pm 0.043 	&  9.919 \pm 0.215 	\\ 
		GH10 	& 10 	& 7.082 \pm 0.038	&  8.035 \pm 0.193 	\\
		GH15 	& 15 	& 6.944 \pm 0.048	&  8.224 \pm 0.188	\\
		GH20 	& 20 	& 6.601 \pm 0.058	&  7.406 \pm 0.193	\\
		\bottomrule
	\end{tabular}
	\caption{The average root-mean-square error.}
	\label{tab:rmse}
\end{table}
As evidenced by the results in Table~\ref{tab:rmse}, the Bayesian quadrature achieves superior RMSE performance for all sigma-point sets. 
In the classical quadrature case the performance improves with increasing number of sigma-points used. 
\begin{table} 
	\centering
	\begin{tabular}{ l l >{$\,}l<{\:$} >{$\:}l<{\,$} }
		\toprule
		\text{Point set} & N & \text{GPQ} & \text{Classical} 	\\ 
		\midrule
		SR 		& 2 	& 3.328 \pm 0.026	& 56.570 \pm 2.728	\\
		UT 		& 3 	& 4.970 \pm 0.343 	&  5.306 \pm 0.481 	\\
		GH5 	& 5 	& 4.088 \pm 0.064 	& 14.722 \pm 0.829 	\\
		GH7 	& 7 	& 4.045 \pm 0.017 	& 12.395 \pm 0.855 	\\
		GH10 	& 10 	& 3.530 \pm 0.012 	&  7.565 \pm 0.534 	\\
		GH15 	& 15 	& 3.468 \pm 0.014	&  7.142 \pm 0.557	\\
		GH20 	& 20 	& 3.378 \pm 0.017 	&  5.664 \pm 0.488	\\
		\bottomrule
	\end{tabular}
	\caption{The average negative log-likelihood.}
	\label{tab:nll}
\end{table}
Table \ref{tab:nll} shows that the performance of GPQKF is clearly superior in terms of NLL, which indicates that the estimates produced by the GPQ-based filters are better representations of the unknown true state development.
\begin{table} 
	\centering
	\smallskip
	\begin{tabular}{ l l >{$\,}l<{\:$} >{$\:}l<{\,$} }
		\toprule
		\text{Point set} & N & \text{GPQ} & \text{Classical} \\ 
		\midrule
		SR 		& 2 	& 1.265 \pm 0.010	& 18.585 \pm 0.045	\\
		UT 		& 3 	& 0.363 \pm 0.108 	&  0.897 \pm 0.088 	\\
		GH5 	& 5 	& 4.549 \pm 0.013 	&  9.679 \pm 0.068 	\\
		GH7 	& 7 	& 4.638 \pm 0.006 	&  8.409 \pm 0.076 	\\
		GH10 	& 10 	& 2.520 \pm 0.006 	&  5.315 \pm 0.058 	\\
		GH15 	& 15 	& 2.331 \pm 0.008	&  5.424 \pm 0.059	\\
		GH20 	& 20 	& 1.654 \pm 0.007	&  4.105 \pm 0.055	\\
		\bottomrule
	\end{tabular}
	\caption{The average inclination indicator.}
	\label{tab:nci}
\end{table}
The self-assessment of the filter performance is more credible in the case of GPQ, as indicated by lower inclination \( \nu \) in the Table~\ref{tab:nci}. 
This indicates that the GPQ-based filters are more conservative in their covariance estimates - a consequence of including additional uncertainty (integral variance), which the classical quadrature-based filters do not employ. 
Also note, that the variance of all the evaluated criteria for GPQ-based filters is mostly an order of magnitude lower.

To achieve competitive results, the kernel lengthscale $\ell$ had to be manually set for each filter separately. 
This was done by running the filters with increasing lengthscale, plotting the performance metrics and choosing the value which gave the smallest RMSE and the inclination closest to zero. 
Figure~\ref{fig:hypers_sensitivity} illustrates the effect of changing lengthscale on the overall performance of the GPQKF with UT sigma-points.
\begin{figure} 
	\centering
	\input{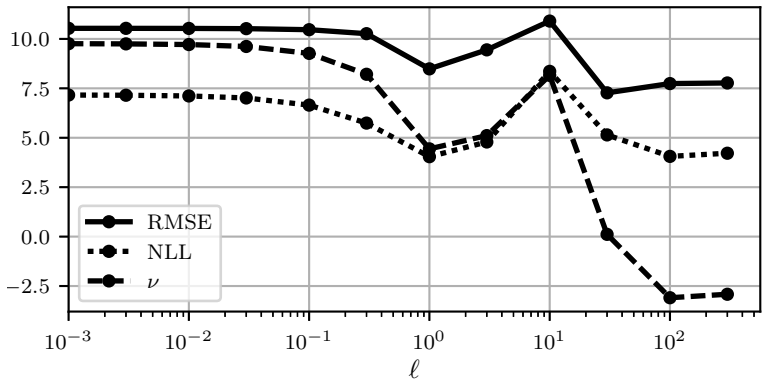}
	\caption{Sensitivity of GPQKF performance (using UT sigma-points) to changes in the lengthscale parameter $\ell$. The choice $ \ell=3 $ minimizes RMSE and yields nearly optimal inclination \( \nu \).}
	\label{fig:hypers_sensitivity}
\end{figure}

\subsection{Target Tracking}

As a more application oriented example, we considered a target tracking scenario adopted from \citep{Athans1968,Julier2000}.
A spherical object falls down from high altitude entering the Earth's atmosphere with high velocity.
The nonlinear dynamics is described by the following set of differential equations
\begin{align}
	\dot{p}(t) &= -v(t) + q_1(t), \label{eq:target_dyn_pos} \\
	\dot{v}(t) &= -v^2(t)\theta(t) e^{-\gamma p(t)} + q_2(t), \\
	\dot{\theta}(t) &= q_3(t), \label{eq:target_dyn_theta}
\end{align}
where \( \gamma = 0.164 \) is a constant and the system state \( \vb{x} = \bmqty{p & v & \theta} \) consists of position (altitude) \( p \), velocity \( v \) and a constant ballistic parameter \( \theta \).
The zero-mean state noise is characterized by \( \E{q_i(t)q_j(s)} = \stNoiseCov\delta(t-s)\delta(i-j) \), where \( \stNoise = \bmqty{q_1 & q_2 & q_3} \).
The range measurements are produced at discrete time intervals by a radar positioned at the altitude of \SI{30}{km} and \SI{30}{km} horizontally to the vertical path of the falling object.
Thus the observation model is
\begin{equation}\label{eq:target_observation_model}
	y_k = \sqrt{s_x^2 + (s_y - p_k)^2} + r_k,
\end{equation}
where \( (s_x,\ s_y) \) is the radar position. 
The measurements were generated with frequency \SI{10}{Hz} and the measurement noise is zero-mean with variance \( \sigma^2_y = \SI{9.2903e-4}{km^2} \).
The mean and covariance of the system initial condition were set to
\begin{align}\label{eq:init_true}
	\vb{x}_0 &= \bmqty{\SI{90}{km} & \SI{6}{km.s^{-1}} & 1.5} \\
	\mb{P}_0 &= \diag{\bmqty{\SI{0.0929}{km^2} & \SI{1.4865}{km^2.s^{-2}} & 10^{-4}}}
\end{align}
while the filter used different initial state estimate
\begin{align}\label{eq:init_model}
	\stMean_{0|0} &= \bmqty{\SI{90}{km} & \SI{6}{km.s^{-1}} & 1.7} \\
	\stCov_{0|0} &= \diag{\bmqty{\SI{0.0929}{km^2} & \SI{1.4865}{km^2.s^{-2}} & 10}}
\end{align}
which implies a lighter object than in reality. 

In the experiments, we focused on the comparison of our GPQKF with the UT points and the UKF, because this filter was previously used in \citep{Julier2000} to demonstrate its superiority over the EKF on the same tracking problem.
The parameters of the UKF were set to \( \kappa = 0,\ \alpha =1,\ \beta = 2 \) following the recommended heuristics \citep{Saerkkae2013}. 
The GPQKF used different kernel parameters for the dynamics, \( \alpha_f = 0.5,\ \ell_f = \bmqty{10 & 10 & 10} \), and the measurement nonlinearity, \( \alpha_h = 0.5,\ \ell_h = \bmqty{15 & 20 & 20} \).
All filters operated with a discrete-time model obtained by Euler approximation with step size \( \Delta t = \SI{0.1}{s} \).
The discretized model is given by
\begin{align}
	p(k+1) 		&= p(k) - \Delta t\, v(k) + q_1(k), \label{eq:target_dyn_pos_disc} \\
	v(k+1) 		&= v(k) - \Delta t\, v^2(k)\theta(k)\exp(-\gamma p(k)) + q_2(k), \\
	\theta(k+1) &= \theta(k) + q_3(k) \label{eq:target_dyn_theta_disc}.
\end{align}
We generated 100 truth trajectories by simulating the continuous-time dynamics, given by the \crefrange{eq:target_dyn_pos}{eq:target_dyn_theta}, for 30 time steps by 4th-order Runge-Kutta scheme and computed the average RMSE and inclination indicator \( \nu \) for both tested filters.
\Cref{fig:reentry_pos_vel} shows realizations of the altitude and velocity trajectories along with the average trajectory.
Note that when the object is passing directly in front of the radar at approximately \( t=10\ \mathrm{s} \) (i.e. altitude 30 km), the system is almost unobservable.

\begin{figure}[h!]
	\centering
	\input{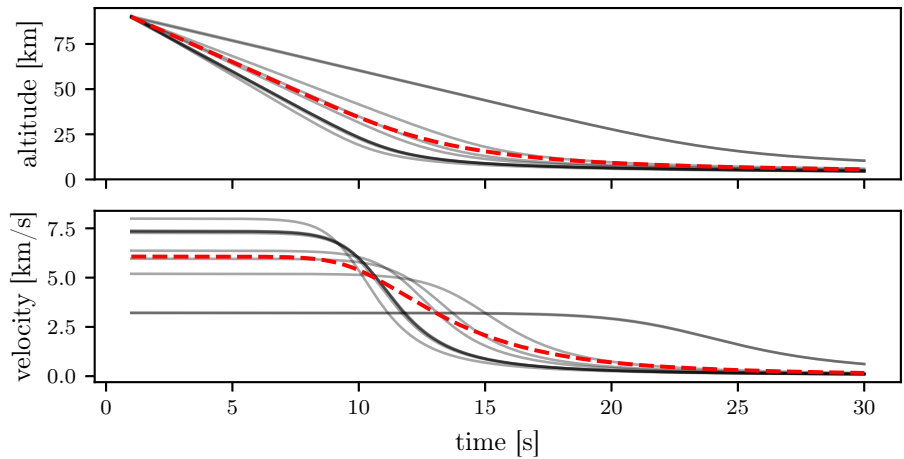}
	\caption{Altitude and velocity evolution in time. Trajectory realizations (black) and the average trajectory (red). The greatest deceleration occurs in the period from 10 to 20 seconds.}
	\label{fig:reentry_pos_vel}
\end{figure}
\Cref{fig:reentry_position_rmse} depicts the RMSE for each time step averaged over trajectory simulations.
The RMSE of the GPQKF tends to be better for all state vector components.
The biggest difference is evident in the RMSE of the ballistic parameter where GPQKF shows significantly better performance during the period of the greatest deceleration.
\begin{figure}[h!]
	\centering
	\input{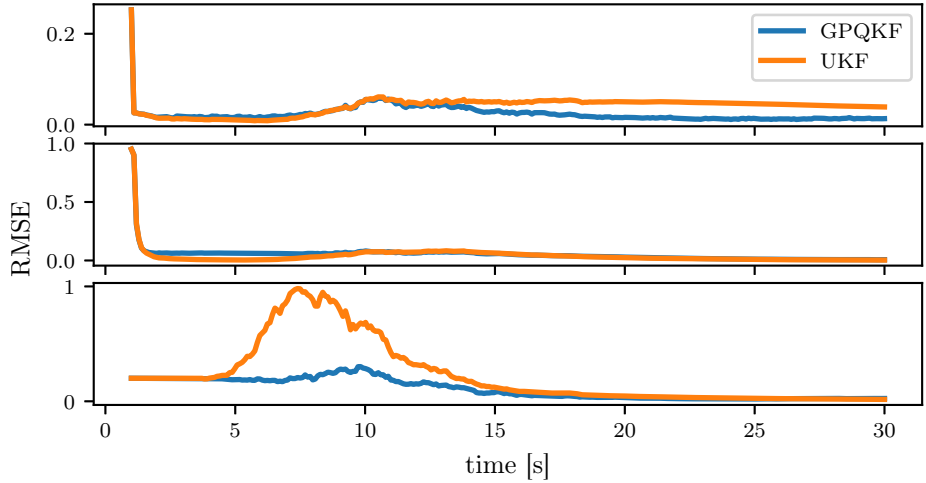}
	\caption{Evolution of the average RMSE in time for the GPQKF with the UT points and the UKF. From top to bottom: position, velocity and ballistic parameter.}
	\label{fig:reentry_position_rmse}
\end{figure}
Overall, the UKF shows signs of an unbalanced estimator as evidenced from \Cref{fig:reentry_position_inclination}, where the inclination \( \nu \) rises significantly above zero, indicating excessive optimism.
The GPQKF manages to stay mostly balanced (\( \nu \) wobbles around zero) with the exception of velocity, where it tips toward pessimism towards the end of the trajectory.
This behaviour is mostly likely caused by the inclusion of additional functional uncertainty in the transformed covariance as shown in \cref{eq:gpq_mt_variance_0,eq:gpq_mt_variance_1}.
\begin{figure}[h!]
	\centering
	\input{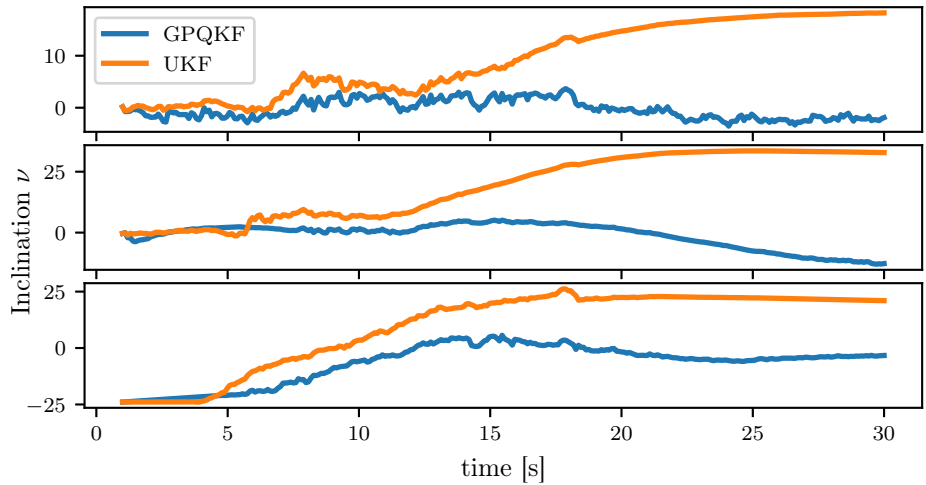}
	\caption{Evolution of the average inclination in time for the GPQKF with the UT points and the UKF. From top to bottom: position, velocity and ballistic parameter.}
	\label{fig:reentry_position_inclination}
\end{figure}




\section{Conclusion}\label{sec:conclusion}
In this article, we have shown how a Bayesian view of quadrature can be leveraged for the design of general purpose moment transform.
Unlike the classical transforms, the proposed GPQ transform accounts for the integration error incurred in computing the mean by inflating the transformed covariance.
The underlying model in the proposed transform is a non-parametric GP regression model, which brings a lot of advantages. 
Namely, the transform is not restricted by polynomial assumptions on the integrand (unlike the classical methods) and it quantifies predictive uncertainty, which eventually translates into integral uncertainty.
The proposed moment transform is entirely general, in that the equations hold for any kernel and input density, however, analytically tractable kernel-density pairs are preferable.
We showed that the transform may outperform classical transforms on a coordinate conversion and two nonlinear sigma-point filtering examples.
In both experiments, the filters based on the GPQ give more realistic estimates of the covariance, hence are better at self-assessing their estimation error.
Currently, the biggest challenge is finding optimal values of the kernel parameters.


\bibliographystyle{plainnat}
\bibliography{refdb}

\begin{thebibliography}{42}
\providecommand{\natexlab}[1]{#1}
\providecommand{\url}[1]{\texttt{#1}}
\expandafter\ifx\csname urlstyle\endcsname\relax
  \providecommand{\doi}[1]{doi: #1}\else
  \providecommand{\doi}{doi: \begingroup \urlstyle{rm}\Url}\fi

\bibitem[Arasaratnam and Haykin(2009)]{Arasaratnam2009}
I.~Arasaratnam and S.~Haykin.
\newblock {Cubature Kalman Filters}.
\newblock \emph{IEEE Transactions on Automatic Control}, 54\penalty0
  (6):\penalty0 1254--1269, 2009.

\bibitem[Athans et~al.(1968)Athans, Wishner, and Bertolini]{Athans1968}
M.~Athans, R.~Wishner, and A.~Bertolini.
\newblock {Suboptimal State Estimation for Continuous-Time Nonlinear Systems
  from Discrete Noisy Measurements}.
\newblock \emph{IEEE Transactions on Automatic Control}, 13\penalty0
  (5):\penalty0 504--514, 1968.
\newblock \doi{10.1109/TAC.1968.1098986}.

\bibitem[Bhar(2010)]{Bhar2010}
R.~Bhar.
\newblock \emph{{Stochastic filtering with applications in finance}}.
\newblock World Scientific, 2010.
\newblock ISBN 978-981-4304-85-6.

\bibitem[{Briol} et~al.(2015){Briol}, {Oates}, {Girolami}, {Osborne}, and
  {Sejdinovic}]{Briol2015}
F.-X. {Briol}, C.~J. {Oates}, M.~{Girolami}, M.~A. {Osborne}, and
  D.~{Sejdinovic}.
\newblock {Probabilistic Integration: A Role for Statisticians in Numerical
  Analysis?}
\newblock \emph{ArXiv e-prints}, December 2015.

\bibitem[Deisenroth et~al.(2009)Deisenroth, Huber, and
  Hanebeck]{Deisenroth2009}
M.~P. Deisenroth, M.~F. Huber, and U.~D. Hanebeck.
\newblock {Analytic moment-based Gaussian process filtering}.
\newblock In \emph{Proceedings of the 26th Annual International Conference on
  Machine Learning - ICML '09}, pages 1--8. ACM Press, 2009.
\newblock \doi{10.1145/1553374.1553403}.

\bibitem[Deisenroth et~al.(2012)Deisenroth, Turner, Huber, Hanebeck, and
  Rasmussen]{Deisenroth2012}
M.~P. Deisenroth, R.~D. Turner, M.~F. Huber, U.~D. Hanebeck, and C.~E.
  Rasmussen.
\newblock {Robust Filtering and Smoothing with Gaussian Processes}.
\newblock \emph{IEEE Transactions on Automatic Control}, 57\penalty0
  (7):\penalty0 1865--1871, 2012.
\newblock \doi{10.1109/TAC.2011.2179426}.

\bibitem[Diaconis(1988)]{Diaconis1988}
P.~Diaconis.
\newblock {Bayesian numerical analysis}.
\newblock In S.~S. Gupta and J.~O. Berger, editors, \emph{Statistical Decision
  Theory and Related Topics IV}, pages 163--175. Springer, 1988.

\bibitem[Dun\'{\i}k et~al.(2013)Dun\'{\i}k, Straka, and \v{S}imandl]{Dunik2013}
J.~Dun\'{\i}k, O.~Straka, and M.~\v{S}imandl.
\newblock {Stochastic Integration Filter}.
\newblock \emph{IEEE Transactions on Automatic Control}, 58\penalty0
  (6):\penalty0 1561--1566, 2013.

\bibitem[Gautschi(2004)]{Gautschi2004}
W.~Gautschi.
\newblock \emph{Orthogonal Polynomials: Computation and Approximation}.
\newblock Numerical Mathematics and Scientific Computation. Oxford University
  Press, 2004.
\newblock ISBN 978-0198506720.

\bibitem[Gelman(2013)]{Gelman2013}
A.~Gelman.
\newblock \emph{{Bayesian Data Analysis}}.
\newblock Chapman and Hall/CRC, 3rd edition, 2013.
\newblock ISBN 978-1439840955.

\bibitem[Gillijns et~al.(2006)Gillijns, Mendoza, Chandrasekar, De~Moor,
  Bernstein, and Ridley]{Gillijns2006}
S.~Gillijns, O.B. Mendoza, J.~Chandrasekar, B.L.R. De~Moor, D.S. Bernstein, and
  A.~Ridley.
\newblock What is the ensemble kalman filter and how well does it work?
\newblock In \emph{American Control Conference, 2006}, page~6, June 2006.
\newblock \doi{10.1109/ACC.2006.1657419}.

\bibitem[Girard et~al.(2003)Girard, Rasmussen, Qui\~{n}onero Candela, and
  Murray-Smith]{Girard2003}
A.~Girard, C.~E. Rasmussen, J.~Qui\~{n}onero Candela, and R.~Murray-Smith.
\newblock {Gaussian Process Priors With Uncertain Inputs Application to
  Multiple-Step Ahead Time Series Forecasting}.
\newblock In S.~Becker, S.~Thrun, and K.~Obermayer, editors, \emph{Advances in
  Neural Information Processing Systems 15}, pages 545--552. MIT Press, 2003.

\bibitem[Gordon et~al.(1993)Gordon, Salmond, and Smith]{Gordon1993}
N.~J. Gordon, D.~J. Salmond, and A.~F.~M. Smith.
\newblock {Novel approach to nonlinear/non-Gaussian Bayesian state estimation}.
\newblock \emph{IEE Proceedings F (Radar and Signal Processing)}, 140\penalty0
  (2):\penalty0 107--113, 1993.

\bibitem[Grewal et~al.(2007)Grewal, Weill, and Andrews]{Grewal2007}
M.~S. Grewal, L.~R. Weill, and A.~P. Andrews.
\newblock \emph{{Global Positioning Systems, Inertial Navigation, and
  Integration}}.
\newblock Wiley, 2007.
\newblock ISBN 978-0-470-09971-1.

\bibitem[Heine et~al.(2006)Heine, Kawohl, and King]{Heine2006}
T.~Heine, M.~Kawohl, and R.~King.
\newblock {Robust Model Predictive Control using the Unscented Transformation}.
\newblock In \emph{2006 IEEE Conference on Computer Aided Control System
  Design}, pages 224--230, Oct 2006.
\newblock \doi{10.1109/CACSD-CCA-ISIC.2006.4776650}.

\bibitem[Horn and Johnson(1990)]{Horn1990}
R.~A. Horn and C.~R. Johnson.
\newblock \emph{Matrix Analysis}.
\newblock Cambridge University Press, 1990.
\newblock ISBN 978-0-521-38632-6.

\bibitem[Ito and Xiong(2000)]{Ito2000}
K.~Ito and K.~Xiong.
\newblock {Gaussian Filters for Nonlinear Filtering Problems}.
\newblock \emph{IEEE Transactions on Automatic Control}, 45\penalty0
  (5):\penalty0 910--927, 2000.
\newblock \doi{10.1109/9.855552}.

\bibitem[Jiang et~al.(2003)Jiang, Sidiropoulos, and Giannakis]{Jiang2003}
T.~Jiang, N.D. Sidiropoulos, and G.B. Giannakis.
\newblock Kalman filtering for power estimation in mobile communications.
\newblock \emph{Wireless Communications, IEEE Transactions on}, 2\penalty0
  (1):\penalty0 151--161, 2003.
\newblock \doi{10.1109/TWC.2002.806386}.

\bibitem[Julier et~al.(2000)Julier, Uhlmann, and Durrant-Whyte]{Julier2000}
S.~J. Julier, J.~K. Uhlmann, and H.~F. Durrant-Whyte.
\newblock {A New Method for the Nonlinear Transformation of Means and
  Covariances in Filters and Estimators}.
\newblock \emph{IEEE Transactions on Automatic Control}, 45\penalty0
  (3):\penalty0 477--482, 2000.

\bibitem[Kitagawa(1996)]{Kitagawa1996}
G.~Kitagawa.
\newblock {Monte Carlo Filter and Smoother for Non-Gaussian Nonlinear State
  Space Models}.
\newblock \emph{Journal of Computational and Graphical Statistics}, 5\penalty0
  (1):\penalty0 1--25, 1996.

\bibitem[Ko et~al.(2007)Ko, Klein, Fox, and Haehnel]{Ko2007}
J.~Ko, D.~J. Klein, D.~Fox, and D.~Haehnel.
\newblock {GP-UKF: Unscented Kalman Filters with Gaussian Process Prediction
  and Observation Models}.
\newblock In \emph{Intelligent Robots and Systems, 2007. IROS 2007. IEEE/RSJ
  International Conference on}, volume~54, pages 1901--1907. IEEE, 2007.
\newblock \doi{10.1109/IROS.2007.4399284}.

\bibitem[Li and Zhao(2006)]{Li2006}
X.~R. Li and Z.~Zhao.
\newblock {Measuring Estimator's Credibility: Noncredibility Index}.
\newblock In \emph{Information Fusion, 2006 9th International Conference on},
  pages 1--8, 2006.
\newblock \doi{10.1109/ICIF.2006.301770}.

\bibitem[McNamee and Stenger(1967)]{McNamee1967}
J.~McNamee and F.~Stenger.
\newblock {Construction of fully symmetric numerical integration formulas}.
\newblock \emph{Numerische Mathematik}, 10\penalty0 (4):\penalty0 327--344,
  1967.
\newblock ISSN 0029599X.
\newblock \doi{10.1007/BF02162032}.

\bibitem[Minka(2000)]{Minka2000}
T.~P. Minka.
\newblock {Deriving Quadrature Rules from Gaussian Processes}.
\newblock Technical report, Statistics Department, Carnegie Mellon University,
  Tech. Rep, 2000.

\bibitem[Oates et~al.(2015)Oates, Osborne, and Girolami]{Oates2015}
C.~J. Oates, M.~A Osborne, and M.~Girolami.
\newblock {Frank-Wolfe Bayesian Quadrature : Probabilistic Integration with
  Theoretical Guarantees}.
\newblock \emph{ArXiv e-prints}, pages 1--19, 2015.
\newblock URL \url{http://arxiv.org/abs/1506.02681}.

\bibitem[O'Hagan(1991)]{OHagan1991}
A.~O'Hagan.
\newblock {Bayes-Hermite quadrature}.
\newblock \emph{Journal of Statistical Planning and Inference}, 29\penalty0
  (3):\penalty0 245--260, 1991.
\newblock \doi{10.1016/0378-3758(91)90002-V}.

\bibitem[Osborne and Garnett(2012)]{Osborne2012}
M.~A. Osborne and R.~Garnett.
\newblock {Bayesian Quadrature for Ratios}.
\newblock In \emph{International Conference on Artificial Intelligence and
  Statistics}, pages 832--840, 2012.

\bibitem[Osborne et~al.(2012)Osborne, Rasmussen, Duvenaud, Garnett, and
  Roberts]{Osborne2012a}
M.~A. Osborne, C.~E. Rasmussen, D.~K. Duvenaud, R.~Garnett, and S.~J. Roberts.
\newblock {Active Learning of Model Evidence Using Bayesian Quadrature}.
\newblock In \emph{Advances in Neural Information Processing Systems (NIPS)},
  pages 46--54, 2012.

\bibitem[Poincar\'{e}(1896)]{Poincare1896}
H.~Poincar\'{e}.
\newblock \emph{{Calcul des probabilit\'{e}s}}.
\newblock Paris, Gauthier-Villars, 1896.

\bibitem[Pr\"{u}her and \v{S}imandl(2016)]{Prueher2016}
J.~Pr\"{u}her and M.~\v{S}imandl.
\newblock {Bayesian Quadrature Variance in Sigma-point Filtering}.
\newblock In J.~Filipe, O.~Gusikhin, Madani K., and J.~Sasiadek, editors,
  \emph{Informatics in Control, Automation and Robotics, 12th International
  Conference, ICINCO 2015 Colmar, Alsace, France, 21-23 July, 2015 Revised
  Selected Papers}, volume 370 of \emph{Lecture Notes in Electrical
  Engineering}. Springer International Publishing, 2016.

\bibitem[Rasmussen and Ghahramani(2003)]{Rasmussen2003a}
C.~E. Rasmussen and Z.~Ghahramani.
\newblock {Bayesian Monte Carlo}.
\newblock In S.~Becker, S.~Thrun, and K.~Obermayer, editors, \emph{Advances in
  Neural Information Processing Systems 15}, number~1, pages 505--512. MIT
  Press, 2003.

\bibitem[Rasmussen and Williams(2006)]{Rasmussen2006}
C.~E. Rasmussen and C.~K. Williams.
\newblock \emph{{Gaussian Processes for Machine Learning}}.
\newblock The MIT Press, 2006.
\newblock ISBN 978-0-262-18253-9.

\bibitem[Ross et~al.(2015)Ross, Proulx, and Karpenko]{Ross2015}
I.~M. Ross, R.~J. Proulx, and M.~Karpenko.
\newblock Unscented guidance.
\newblock In \emph{2015 American Control Conference (ACC)}, pages 5605--5610,
  July 2015.
\newblock \doi{10.1109/ACC.2015.7172217}.

\bibitem[Sandblom and Svensson(2012)]{Sandblom2012}
F.~Sandblom and L.~Svensson.
\newblock {Moment Estimation Using a Marginalized Transform}.
\newblock \emph{IEEE Transactions on Signal Processing}, 60\penalty0
  (12):\penalty0 6138--6150, 2012.

\bibitem[S\"{a}rkk\"{a}(2013)]{Saerkkae2013}
S.~S\"{a}rkk\"{a}.
\newblock \emph{{Bayesian Filtering and Smoothing}}.
\newblock Cambridge University Press, New York, 2013.
\newblock ISBN 978-1-107-61928-9.

\bibitem[S\"{a}rkk\"{a} et~al.(2014)S\"{a}rkk\"{a}, Hartikainen, Svensson, and
  Sandblom]{Sarkka2014}
S.~S\"{a}rkk\"{a}, J.~Hartikainen, L.~Svensson, and F.~Sandblom.
\newblock {Gaussian Process Quadratures in Nonlinear Sigma-Point Filtering and
  Smoothing}.
\newblock In \emph{Information Fusion (FUSION), 2014 17th International
  Conference on}, pages 1--8, 2014.

\bibitem[{S{\"a}rkk{\"a}} et~al.(2016){S{\"a}rkk{\"a}}, {Hartikainen},
  {Svensson}, and {Sandblom}]{Saerkkae2016}
S.~{S{\"a}rkk{\"a}}, J.~{Hartikainen}, L.~{Svensson}, and F.~{Sandblom}.
\newblock {On the relation between Gaussian process quadratures and sigma-point
  methods}.
\newblock \emph{Journal of Advances in Information Fusion}, 11:\penalty0
  31--46, June 2016.
\newblock ISSN 1557-6418.
\newblock URL \url{http://arxiv.org/abs/1504.05994}.

\bibitem[Smith et~al.(1962)Smith, Schmidt, and McGee]{Smith1962}
G.~L. Smith, S.~F. Schmidt, and L.~A. McGee.
\newblock {Application of statistical filter theory to the optimal estimation
  of position and velocity on board a circumlunar vehicle}.
\newblock Technical report, NASA Ames Research Center: Technical Report R-135,
  1962.

\bibitem[Tronarp et~al.(2016)Tronarp, Hostettler, and
  S\"{a}rkk\"{a}]{Tronarp2016}
F.~Tronarp, R.~Hostettler, and S.~S\"{a}rkk\"{a}.
\newblock Sigma-point filtering for nonlinear systems with non-additive
  heavy-tailed noise.
\newblock In \emph{2016 19th International Conference on Information Fusion
  (FUSION)}, pages 1859--1866, July 2016.

\bibitem[Wasserman(2007)]{Wasserman2007}
L.~Wasserman.
\newblock \emph{{All of Nonparametric Statistics}}.
\newblock Springer, 2007.
\newblock ISBN 978-0387251455.

\bibitem[Wu et~al.(2006)Wu, Hu, Wu, and Hu]{Wu2006}
Y.~Wu, D.~Hu, M.~Wu, and X.~Hu.
\newblock {A Numerical-Integration Perspective on Gaussian Filters}.
\newblock \emph{IEEE Transactions on Signal Processing}, 54\penalty0
  (8):\penalty0 2910--2921, 2006.

\bibitem[Zangl and Steiner(2008)]{Zangl2008}
H.~Zangl and G.~Steiner.
\newblock {Optimal Design of Multiparameter Multisensor Systems}.
\newblock \emph{IEEE Transactions on Instrumentation and Measurement},
  57\penalty0 (7):\penalty0 1484--1491, July 2008.
\newblock ISSN 0018-9456.
\newblock \doi{10.1109/TIM.2008.917175}.

\end{thebibliography}
\end{document}